\documentclass{article}

\usepackage[english]{babel}

\usepackage[a4paper,top=2cm,bottom=2cm,left=3cm,right=3cm,marginparwidth=1.75cm]{geometry}

\usepackage[bookmarks=true,hypertexnames=false,pagebackref]{hyperref}
\hypersetup{colorlinks=true, citecolor=blue, linkcolor=red, urlcolor=blue}
\usepackage[textsize=tiny]{todonotes}
\usepackage{mathrsfs}
\usepackage{graphicx}

\usepackage{amsthm,amsmath,amssymb}
\usepackage{xifthen}
\usepackage[ruled,linesnumbered,vlined]{algorithm2e}

\usepackage{cleveref} % Load algorithm2e before cleveref in order to ref algorithms correctly
\usepackage{cite}
\newtheorem{claim}{Claim}[section]

\newtheorem{theorem}{Theorem}[section]

\newtheorem{lemma}[theorem]{Lemma}

\newtheorem{definition}[theorem]{Definition}
\newtheorem{proposition}[theorem]{Proposition}

\newtheorem*{remark}{Remark}

\def\*#1{\mathbf{#1}} \def\+#1{\mathcal{#1}} \def\-#1{\mathrm{#1}}\def\^#1{\mathbb{#1}}\def\!#1{\mathtt{#1}}\def\@#1{\mathscr{#1}}
\newcommand{\set}[1]{\left\{#1\right\}}

\renewcommand{\mid}{\;\middle\vert\;}

\renewcommand{\Pr}[2][]{ \ifthenelse{\isempty{#1}}
  {\mathbf{Pr}\left[#2\right]} {\mathbf{Pr}_{#1}\left[#2\right]} }
\newcommand{\E}[2][]{ \ifthenelse{\isempty{#1}}
  {\mathbf{E}\left[#2\right]}
  {\mathbf{E}_{#1}\left[#2\right]} }
\newcommand{\Var}[2][]{ \ifthenelse{\isempty{#1}}
  {\mathbf{Var}\left[#2\right]}
  {\mathbf{Var}_{#1}\left[#2\right]} }
\newcommand{\Ent}[2][]{ \ifthenelse{\isempty{#1}}
  {\mathbf{Ent}\left[#2\right]}
  {\mathbf{Ent}_{#1}\left[#2\right]} }
\renewcommand{\emptyset}{\varnothing}

\allowdisplaybreaks
\newcommand{\alg}{\mathsf{ALG}}
\newcommand{\opt}{\mathsf{OPT}}
\newcommand*\samethanks[1][\value{footnote}]{\footnotemark[#1]}

\title{Improved Competitive Ratio for Edge-Weighted Online Stochastic Matching\thanks{The authors are ordered alphabetically.}}
\author{Yilong Feng
	\thanks{IOTSC, University of Macau. \{mc15517,xiaoweiwu,yc17423\}@um.edu.mo.}
        \and Guoliang Qiu \thanks{John Hopcroft Center for Computer Science, Shanghai Jiao Tong University. guoliang.qiu@sjtu.edu.cn. Part of this work was done while the author was visiting the University of Macau.}
	\and Xiaowei Wu\samethanks[2]
	\and Shengwei Zhou\samethanks[2]}

\begin{document}
\maketitle

\begin{abstract}
    We consider the edge-weighted online stochastic matching problem, in which an edge-weighted bipartite graph $G=(I\cup J, E)$ with offline vertices $J$ and online vertex types $I$ is given.
    The online vertices have types sampled from $I$ with probability proportional to the arrival rates of online vertex types.
    The online algorithm must make immediate and irrevocable matching decisions with the objective of maximizing the total weight of the matching.
    For the problem with general arrival rates, Feldman et al. (FOCS 2009) proposed the \textsf{Suggested Matching} algorithm and showed that it achieves a competitive ratio of $1-1/e \approx 0.632$.
    The ratio has recently been improved to $0.645$ by Yan (2022), who proposed the \textsf{Multistage Suggested Matching} (MSM) algorithm.
    In this paper, we propose the \textsf{Evolving Suggested Matching} (ESM) algorithm and show that it achieves a competitive ratio of $0.650$.
\end{abstract}

\section{Introduction}\label{sec:intro}

Motivated by its real-world applications, the online bipartite matching problem has received extensive attention since the work of Karp, Vazirani, and Vazirani~\cite{conf/stoc/KarpVV90} in 1990.
The problem is defined on a bipartite graph where one side of the vertices are given to the algorithm in advance (aka. the \emph{offline} vertices), and the other side of (unknown) vertices arrive online one by one.
Upon the arrival of an online vertex, its incident edges are revealed and the online algorithm must make immediate and irrevocable matching decisions with the objective of maximizing the size of the matching.
The performance of the online algorithm is measured by the \emph{competitive ratio}, which is the worst ratio between the size of matching computed by the online algorithm and that of the maximum matching, over all online instances. 
As shown by Karp et al.~\cite{conf/stoc/KarpVV90}, the celebrated \textsf{Ranking} algorithm achieves a competitive ratio of $1-1/e \approx 0.632$ and this is the best possible for the problem.
However, the assumption that the algorithm has no prior information
% \qtodo{it seems that ``zero information'' is not an appropriate expression. simple write it as no information?} 
regarding the online vertices and the adversary decides the arrival order of these vertices, is believed to be too restrictive and in fact unrealistic.
Therefore, several other arrival models with weaker adversaries, including the random arrival model~\cite{conf/stoc/MahdianY11,conf/stoc/KarandeMT11,journals/talg/HuangTWZ19,conf/wine/JinW21}, the degree-bounded model~\cite{journals/teco/NaorW18,conf/soda/CohenW18,conf/esa/AlbersS22} and the stochastic model~\cite{conf/focs/FeldmanMMM09,journals/mor/JailletL14,conf/stoc/HuangS21,conf/stoc/0002SY22},
% \qtodo{weird, the article 16 is 2022, and 9,19 are both old. Is there any reason for this citation?} 
have been proposed.

The stochastic model is proposed by Feldman et al.~\cite{conf/focs/FeldmanMMM09}, in which the arrivals of online vertices follow a known distribution.
Specifically, in the stochastic setting there is a bipartite graph $G = (I\cup J, E)$ that is known by the algorithm, where $J$ contains the offline vertices and $I$ contains the online vertex \emph{types}, where each vertex type $i\in I$ is associated with an arrival rate $\lambda_i$.
There are $\Lambda = \sum_{i\in I} \lambda_i$ online vertices to be arrived.
Each online vertex has a type sampled from $I$ independently, where type $i\in I$ is sampled with probability $\lambda_i/\Lambda$.
The online vertex with type $i$ has its set of neighbors defined by the neighbors of $i$ in the graph $G = (I\cup J, E)$.
The competitive ratio for the online algorithm is then measured by the worst ratio of $\frac{\E{\alg}}{\E{\opt}}$ over all instances, where $\alg$ denotes the size of matching produced by the algorithm and $\opt$ denotes that of the maximum matching.
For the online stochastic matching problem, Feldman et al.~\cite{conf/focs/FeldmanMMM09} proposed the \textsf{Suggested Matching} algorithm that achieves a competitive ratio of $1-1/e$ and the \textsf{Two Suggested Matching} that is $0.67$-competitive for instances with integral arrival rates.
These algorithms are based on the framework that makes matching decisions in accordance to some offline optimal solution $\boldsymbol{x}$ pre-computed on the instance $(G,\{ \lambda_i \}_{i\in I})$.
Especially, the \textsf{Two Suggested Matching} algorithm employs a novel application of the idea called \emph{the power of two choices} by specifying two offline neighbors and matching one of them if any of these two neighbors is unmatched.
% two offline neighbors sampled with probability proportional to the corresponding variables in the solution $\boldsymbol{x}$ when making the matching decision. 
% \qtodo{not sure for this, the specified rules in feldman seems not like this statement. should be: by specifying two offline neighbors in some way and matching one of them if any of these two neighbors is unmatched?}
The approximation ratio was later improved by a sequence of works~\cite{conf/esa/BahmaniK10,journals/mor/ManshadiGS12,journals/mor/JailletL14,journals/algorithmica/BrubachSSX20a}, resulting in the state-of-the-art competitive ratio $0.7299$ by~\cite{journals/algorithmica/BrubachSSX20a} under the integral arrival rate assumption.
Without this assumption, the first competitive ratio beating $1-1/e$ was obtained by Manshadi et al.~\cite{journals/mor/ManshadiGS12}, who provided a $0.702$-competitive algorithm for the problem.
The competitive ratio was then improved to $0.706$, $0.711$, and $0.716$ by Jaillet and Lu~\cite{journals/mor/JailletL14}, Huang and Shu~\cite{conf/stoc/HuangS21} and Huang et al.~\cite{conf/stoc/0002SY22}, respectively. 
Notably, Huang and Shu~\cite{conf/stoc/HuangS21} established the asymptotic equivalence between the original stochastic arrival model and the Poisson arrival model in which online vertex types arrive independently following Poisson processes. 

The weighted versions of the online stochastic matching problem have also received a considerable amount of attention.
In the edge-weighted (resp. vertex-weighted) version of the problem, each edge (resp. offline vertex) is associated with a non-negative weight and the objective is to compute a matching with maximum total edge (resp. offline vertex) weight. 
For the vertex-weighted version, Jaillet and Lu~\cite{journals/mor/JailletL14} and Brubach et al.~\cite{journals/algorithmica/BrubachSSX20a} achieved a competitive ratio of $0.725$ and $0.7299$, respectively, under the integral arrival rate assumption.
Without this assumption, Huang and Shu~\cite{conf/stoc/HuangS21} and Tang et al.~\cite{conf/stoc/TangWW22} achieved a competitive ratio of $0.7009$ and $0.704$ respectively, while the state-of-the-art ratio $0.716$ was achieved by Huang et al.~\cite{conf/stoc/0002SY22}.
For the edge-weighted version, under the integral arrival rate assumption, competitive ratios $0.667$ and $0.705$ were proved by Haeuper et al.~\cite{conf/wine/HaeuplerMZ11} and Brubach et al.~\cite{journals/algorithmica/BrubachSSX20a}, respectively.
Without the assumption, the $(1-1/e)$-competitive \textsf{Suggested Matching} algorithm by Feldman et al~\cite{conf/focs/FeldmanMMM09} remained the state-of-the-art until recently Yan~\cite{Y22} proposed an algorithm called \textsf{Multistage Suggested Matching} (MSM) algorithm and showed that it achieves a competitive ratio of $0.645$.
Regarding hardness results, Huang et al.~\cite{conf/stoc/0002SY22} proved that no algorithm can be $0.703$-competitive for the edge-weighted online stochastic matching problem.
The hardness result separates the edge-weighted version of the problem from the unweighted and vertex-weighted versions, for both of which competitive ratios strictly larger than $0.703$ have already been proved.
It also separates the problem without the integral arrival rate assumption from that with the assumption, which indicates the difficulty of the problem.

\subsection{Our Contribution}

In this paper, we consider the edge-weighted online stochastic matching problem without the integral arrival rate assumption, and propose the \textsf{Evolving Suggested Matching} (ESM) algorithm that improves the state-of-the-art competitive ratio to $0.650$.

\begin{theorem} \label{theo:competitive-ratio-origin}
    The \textsf{Evolving Suggested Matching} algorithm is $0.650$-competitive for edge-weighted online stochastic matching with a sufficiently large number of arrivals.
\end{theorem}

\begin{remark}
   The result follows from the asymptotic equivalence between the stochastic arrival model and the Poisson arrival model established by Huang and Shu\cite{conf/stoc/HuangS21}. For simplicity, we analyze the performance of our ESM algorithm under the Possion arrival model (the formal definition can be found in \Cref{sec:prelim}) instead.
\end{remark}

Our work follows the common framework that formulates the matching in the graph $G=(I\cup J,E)$ into a Linear Program (LP) and uses the corresponding pre-computed optimal solution $\boldsymbol{x}$ to guide the design of our online algorithm.
It can be shown that if we can design an algorithm that matches each edge $(i,j)\in E$ with probability at least $\alpha\cdot x_{ij}$, then the algorithm is $\alpha$-competitive. 
The LP is customized for different problems, and we use the LP proposed by Jaillet and Lu~\cite{journals/mor/JailletL14} in this paper.
By a reduction from~\cite{Y22}, it suffices to consider some kernel instances in which all online vertex types have degree at most two.
We call an online vertex type with one neighbor a \emph{first-class} type and an online vertex type with two neighbors a \emph{second-class} type.
The \textsf{Multistage Suggested Matching} (MSM) algorithm proposed by Yan~\cite{Y22} is a hybrid of the \textsf{Suggested Matching} and the \textsf{Two-Choice} algorithms. In the MSM algorithm, the second-class online vertices follow different strategies at three different stages of the algorithm. 
Inspired by the MSM algorithm, we propose the \textsf{Evolving Suggested Matching} (ESM) algorithm by introducing an activation function that allows us to have theoretically infinitely many different ``stages'' in the algorithm.

\medskip
\noindent
\textbf{Technical Contribution 1: Evolving Strategies by Activation Function}.
With the activation function, our algorithm is able to evolve the matching strategies in time horizon. As discussed in \Cref{sec:ESM}, the activation function controls how ``aggressive'' the second-class online vertex types propose to their neighbors and is general enough to capture many of the existing algorithms for the online stochastic matching problem, including the \textsf{Suggested Matching} algorithm, the \textsf{Two-Choice} algorithm and that of Yan~\cite{Y22}. 
Furthermore, the design of a better online algorithm can thus be reduced to the design of the activation function in a more tractable way, which is crucial to refining the state-of-the-art competitive ratio.

% it enables us to refine the previous analysis of the MSM algorithm using a more flexible and carefully designed activation function.

\smallskip
\noindent
\textbf{Technical Contribution 2: Fine-grained Correlation Analysis.}
One particular difficulty in competitive analysis by introducing the activation function is that most of the matching events are now intricately correlated (as opposed to the \textsf{Suggested Matching} algorithm~\cite{conf/focs/FeldmanMMM09}), and in order to fully utilize the power of the activation function (as opposed to the analysis for the MSM algorithm~\cite{Y22}), we need to carefully bound the matching probability of a vertex or an edge conditioned on the matching status of other vertices or edges. Essentially, all analysis for the two proposals algorithm needs to take into account the failure of the first proposal affected by the pairwise correlations. Existing works~\cite{conf/stoc/HuangS21,Y22} get around this issue by specific relaxation of the correlated events and thus incur some intrinsic loss. 

To improve the competitive analysis, it is inevitable to conduct a fine-grained correlation analysis of the matching probability of a vertex or an edge. As shown in our analysis (in Section~\ref{sec:ESM}), we conjecture that the offline vertices being unmatched up to some time $t$ are positively correlated. However, it is difficult to capture the correlation between the matching events of offline vertices as there are many places where the random decisions are dependent. Instead, we observe that by abstracting the independent arrivals of the ``extended online vertices types'', we can prove a pseudo-positive correlated inequality available for our analysis. We believe this observation is technically interesting and may lead to further inspiration in future works.

\subsection{Other Related Works}

% Due to the extensive literature on the online bipartite matching problem, we are not able to cover all but only the most relevant works.
For the online bipartite matching problem under the adversarial arrival model, the optimal competitive ratio $1-1/e$ was proved in~\cite{journals/sigact/BirnbaumM08,conf/soda/GoelM08,conf/soda/AggarwalGKM11,conf/soda/DevanurJK13} using different analysis techniques (even in the vertex-weighted version). For the problem under random arrivals, it is assumed that the adversary decides the underlying graph but the online vertices arrive following a uniformly-at-random chosen order.
For the unweighted version of this model, the competitive ratios $0.653$ and $0.696$ have been proved by Karande et al.~\cite{conf/stoc/KarandeMT11} and Mahdian and Yan~\cite{conf/stoc/MahdianY11} respectively.
For the vertex-weighted version of the problem, the competitive ratios $0.653$ and $0.662$ were proved by Huang et al.~\cite{journals/talg/HuangTWZ19} and Jin and Williamson~\cite{conf/wine/JinW21}, respectively.
It is well-known that there exists no competitive algorithm for the edge-weighted online bipartite matching problem. Consequently, the edge-weighted version of the problem has attracted attention when additional assumptions are considered. One significant variant is often referred to as the edge-weighted online bipartite matching problem with free disposal, where each offline vertex can be matched repeatedly, but only the heaviest edge matched to it contributes to the objective. 
For the edge-weighted version of the problem with free disposal under the adversarial arrival order, Fahrbach et al.~\cite{journals/jacm/FahrbachHTZ22} have demonstrated a competitive ratio of $0.5086$, which surpasses the long-standing barrier of $0.5$ achieved by the greedy algorithm. Later, Huang et al.~\cite{conf/stoc/0002SY22} further improved the competitive ratio for this problem to $0.706$ in the stochastic setting. Given the significance of the online bipartite matching problem, recent years have witnessed the exploration of more general arrival models beyond the adversarial and stochastic vertex arrival models. These include the general vertex arrival model~\cite{conf/focs/GamlathKMSW19,conf/icalp/WangW15}, the fully online model~\cite{journals/jacm/HuangKTWZZ20,conf/focs/0002T0020}, and the edge arrival model~\cite{conf/focs/GamlathKMSW19,conf/icalp/GravinTW21}.

\section{Preliminaries} \label{sec:prelim}

We consider the edge-weighted online stochastic matching problem.
An instance of the problem consists of a bipartite graph $G = (I \cup J,E)$, a weight function $w$, and the arrival rates $\{ \lambda_i \}_{i\in I}$ of the online vertex types.
In the bipartite graph $G = (I\cup J,E)$, $I$ denotes the set of online vertex types and $J$ denotes the set of offline vertices.
The set $E\subseteq I\times J$ contains the edges between $I$ and $J$, where each edge $(i,j)\in E$ has a non-negative weight $w_{ij}$.
In the stochastic model, each online vertex type $i\in I$ has an arrival rate $\lambda_i$ and $\Lambda=\sum_{i\in I} \lambda_i$.
Online vertices arrive one by one and each of them draws its type $i$ with probability $\frac{\lambda_i}{\Lambda}$ independently.
Any online algorithm must make an immediate and irrevocable matching decision upon the arrival of each online vertex, with the goal of maximizing the total weight of the matching, subject to the constraint that each offline vertex can be matched at most once. 
Throughout, we use $\opt$ to denote the weight of the maximum weighted matching of the realized instance; and $\alg$ to denote the weight of the matching produced by the online algorithm.
Note that both $\alg$ and $\opt$ are random variables, where the randomness of $\opt$ comes from the random realization of the instance while that of $\alg$ comes from both the realization and the random decisions by the algorithm.
The competitive ratio of the algorithm is measured by the infimum of $\frac{\E{\alg}}{\E{\opt}}$ over all problem instances.

\paragraph{Poisson Arrival Model.} 
Instead of fixing the number of online vertices, in the Poisson arrival model, the online vertex of each type $i$ arrives independently following a Poisson process with time horizon $[0,1]$ and arrival rate $\lambda_i$. 
The independence property allows a more convenient competitive analysis.
In this paper, we consider the problem under the Poisson arrival model. Specifically, we show that
\begin{theorem}\label{theo:competitive-ratio}
     The \textsf{Evolving Suggested Matching} algorithm is $0.650$-competitive for edge-weighted online stochastic matching under the Poisson arrival model.
\end{theorem}
Together with the asymptotic equivalence analysis in~\cite{conf/stoc/HuangS21}, \Cref{theo:competitive-ratio} implies \Cref{theo:competitive-ratio-origin}.

\paragraph{Jaillet-Lu LP.}
We use the following linear program $\text{LP}_{\text{JL}}$ proposed by Jaillet and Lu~\cite{journals/mor/JailletL14} to bound the expected offline optimal value for instance $(G,w,\{ \lambda_i \}_{i\in I})$:
\begin{align*}
    \text{maximize} \quad & \sum_{(i,j)\in E} w_{ij}\cdot x_{ij} \\
    \text{subject to} \quad & \sum_{j : (i,j)\in E} x_{ij} \leq \lambda_i, & \forall i\in I ; \\
    \quad & \sum_{i : (i,j)\in E} x_{ij} \leq 1, & \forall j\in J ; \\
    \quad & \sum_{i : (i,j)\in E} \max\{2 x_{ij}-\lambda_i, 0\} \leq 1-\ln 2, & \forall j\in J ; \\
    & x_{ij}\geq 0, \quad & \forall (i,j) \in E . 
\end{align*}

We use $x_i$ and $x_j$ to denote $\sum_{j: (i,j) \in E} x_{i j}$ and $\sum_{i: (i,j) \in E} x_{i j}$, respectively.
Note that for any feasible solution $\boldsymbol{x}$ we have $x_i \leq \lambda_i$ for all $i\in I$ and $x_j \leq 1$ for all $j\in J$. We remark that although the third set of constraints are not linear, they can be transformed into an LP problem by applying the standard technique of introducing auxiliary variables.

% Hence the first and second constraints are trivially valid for the offline optimal value $\opt$, while the third constraint has been proved in~\cite{conf/stoc/HuangS21} by the converse of Jensen’s inequality.
% We briefly restate the property that the optimal solution of $\text{LP}_{\text{JL}}$ can bound the expected offline optimal value.

\begin{lemma}[Analysis Framework] \label{lemma:analysis-framework}
    If an online algorithm matches each edge $(i,j)\in E$ with probability at least $\alpha\cdot x_{i j}$ for any instance $(G,w, \{ \lambda_i \}_{i\in I})$, where $\boldsymbol{x}$ is the optimal solution to the above LP, then the algorithm is $\alpha$-competitive under the Poisson arrival model.
\end{lemma}
\begin{proof}
    Since the algorithm matches each edge $(i,j)\in E$ with probability at least $\alpha\cdot x_{i j}$, the expected total weight of the matching produced by the algorithm is
    \begin{align*}
        \E{\alg}=&\sum_{(i,j)\in E}w_{ij}\cdot  \Pr{(i,j) \text{ is matched by the algorithm}}\\
        \geq & \sum_{(i,j)\in E}w_{ij}\cdot \alpha\cdot x_{ij} = \alpha\cdot P^* \geq \alpha \cdot \E{\opt},
    \end{align*}
    where $P^*$ denotes the objective of the optimal solution $\boldsymbol{x}$.
    The last inequality follows because by defining $x'_{ij}$ to be the probability that edge $(i,j)$ is included in the maximum weighted matching of the realized instance, we can obtain a feasible solution\footnote{The proof showing that $\boldsymbol{x}'$ satisfies all constraints (especially the third set of constraints) can be found in the appendix of~\cite{conf/stoc/HuangS21}} with objective $\E{\opt}$. Here, the feasibility of the third set of constraints can be briefly explained as follows: Under the Poisson arrival model, as shown in~\cite{conf/stoc/0002SY22}, it holds that for any offline vertex $j$ and any subset of online vertex types $S$ that are adjacent to $j$, the probability of $j$ getting matched among $S$ is bounded by $1-e^{-\sum_{i\in S}\lambda_i}$. Together with the converse of Jensen’s inequality proposed by ~\cite{conf/stoc/HuangS21,conf/stoc/0002SY22}, the third set of constraints holds.
\end{proof}

Following the above lemma, in the rest of the paper, we focus on lower bounding the minimum of $\frac{1}{x_{i j}}\cdot \Pr{(i,j) \text{ is matched by the algorithm}}$ over all edges $(i,j)\in E$.

\paragraph{Suggested Matching.} 
We first give a brief review of the edge-weighted version of the \textsf{Suggested Matching} algorithm, proposed by Feldman et al.~\cite{conf/focs/FeldmanMMM09}.
The algorithm starts from an optimal solution to a linear program\footnote{The LP used in~\cite{conf/focs/FeldmanMMM09} is similar to the Jaillet-Lu LP defined above but without the third set of constraints.}, in which $x_{i j}$ is the corresponding variable for edge $(i,j)\in E$.
The \textsf{Suggested Matching} algorithm proceeds as follows: when an online vertex of type $i$ arrives, it chooses an offline neighbor $j$ with probability ${x_{i j}}/{\lambda_{i}}$ and propose to $j$.
Note that by the feasibility of solution $\boldsymbol{x}$, the probability distribution is well-defined.
If $j$ is not matched, then the algorithm includes $(i,j)$ into the matching.
It can be shown that each edge $(i,j)$ will be matched with probability at least $(1-1/e)\cdot x_{ij}$ when the algorithm terminates, which implies a competitive ratio of $1-1/e$.

\paragraph{Two-Choice.}
The idea of \emph{power of two choices} was first introduced by Feldman et al.~\cite{conf/focs/FeldmanMMM09}, in the algorithm \textsf{Two Suggested Matching}.
Generally speaking, the idea is to allow each online vertex to choose two (instead of one) neighbors, and match one of them if any of these two neighbors is unmatched.
Formally, the \textsf{Two-Choice} algorithm is described as follows: upon the arrival of each online vertex of type $i\in I$, the algorithm chooses two different offline neighbors $j_1, j_2$ following a distribution defined by some optimal solution to an LP.
If $j_1$ is unmatched, then the algorithm matches $i$ to $j_1$; otherwise, the algorithm matches $i$ to $j_2$ if $j_2$ is unmatched.
We call $j_1$ and $j_2$ the \emph{first-choice} and \emph{second-choice} of $i$, respectively.
As introduced, the \textsf{Two-Choice} algorithm has achieved great success in the research field of online stochastic matching problems.

\section{Multistage Suggested Matching}

As a warm-up, we first briefly review the recent progress on the edge-weighted online stochastic matching problem by Yan~\cite{Y22}. 

\paragraph{Kernel Instances.}
We call an instance consisting of graph $G=(I\cup J, E)$, arrival rates $\{\lambda_i\}_{i\in I}$, and a fractional matching $\boldsymbol{x}$ of $\text{LP}_{\text{JL}}$ a \emph{kernel} instance if (1) there are only two classes of online vertex types: one with a single offline neighbor $j$ such that $x_{ij}=\lambda_i$, and the other with two offline neighbors $j_1, j_2$ such that $x_{i j_1}=x_{i j_2}=\frac{1}{2}\lambda_i$; (2) for any offline vertex $j\in J$, we have $x_{j} = 1$.

\medskip

Yan~\cite{Y22} showed that if there exists an online algorithm on the kernel instances such that for any edge $(i,j)\in E$, the probability that $(i,j)$ gets matched by the algorithm is at least $\alpha\cdot x_{ij}$, then we can transform this algorithm to an $\alpha$-competitive algorithm for general problem instances. Specifically, for any general problem instances with an optimal fractional matching $\boldsymbol{x}$ of $\text{LP}_{\text{JL}}$, we can first assume that $x_i=\lambda_i$ for any online vertex type and $x_j=1$ for any offline vertex. Otherwise, a satisfying instance can be constructed by introducing some dummy online vertex types and offline vertices with specific fractional matching on the involved dummy edges such that the objective and the feasibility of the fractional matching are preserved. To reduce the number of online vertex types, Yan~\cite{Y22} proposed a split scheme to split each online vertex type into subtypes and pair up the fractional matching on their edges systematically for its feasibility. The split scheme can be simulated by the downsampling of the online vertex types.

In the rest of this paper, we only consider the kernel instances. 
We call an edge $(i,j)$ a \emph{first-class} edge if $x_{ij}=\lambda_i$, or a \emph{second-class} edge if $x_{ij}=\frac{\lambda_i}{2}$. 
If the online vertex type has an incident first-class (resp. second-class) edge, we call it a first-class (resp. second-class) online vertex type. 
For each offline vertex $j$, we use $N_1(j)$ and $N_2(j)$ to denote the set of first-class and second-class neighbors of $j$, respectively.
Let $y_j:=\sum_{i \in N_1(j)} x_{ij}$ be the sum of variables corresponding to first-class edges incident to $j$.
Note that for the kernel instance, $y_j$ also denotes the total arrival rate of first-class neighbors of $j$.
By the feasibility of solution $\boldsymbol{x}$ we have:
\begin{equation*}
    y_j \leq 1-\ln 2 , \qquad  \forall j\in J.
\end{equation*}

\paragraph{Multistage Suggested Matching Algorithm.}

While the \textsf{Two-Choice} algorithm provides good competitive ratios for the unweighted and vertex-weighted online stochastic matching problems~\cite{conf/focs/FeldmanMMM09,journals/algorithmica/BrubachSSX20a,conf/stoc/HuangS21}, it cannot be straightforwardly applied to the edge-weighted version of the problem.
Specifically, traditional analysis for the unweighted or vertex-weighted versions of the problem focuses on lower bounding the probability of each offline vertex being matched by the algorithm.
However, for the edge-weighted version we need to lower bound the probability of each edge being matched (see Lemma~\ref{lemma:analysis-framework}), and it is not difficult to construct examples of kernel instances in which some (first-class) edge $(i,j)\in E$ is matched with probability strictly less than $(1-1/e)\cdot x_{ij}$ in the \textsf{Two-Choice} algorithm.
In contrast, the \textsf{Multistage Suggested Matching} (MSM) algorithm proposed by Yan~\cite{Y22} is a hybrid of the \textsf{Suggested Matching} and the \textsf{Two-Choice} algorithms.
Specifically, in the MSM algorithm, the first-class vertices always follow the \textsf{Suggested Matching} algorithm while the second-class vertices follow different strategies at different stages of the algorithm.
In the first stage of the algorithm, all second-class vertices are discarded; in the second stage the second-class vertices follow the \textsf{Suggested Matching} algorithm; and in the last stage they follow the \textsf{Two-Choice} algorithm.
Intuitively speaking, the second-class vertices are getting more and more aggressive in terms of trying to match their neighbors, as time goes by.
The design of the first stage is crucial because without this stage the matching probability of a first-class edge $(i,j)\in E$ with $x_{ij} = 1$ will be at most $1-1/e$.
The design of the third stage is also important because without it the performance of the algorithm will not be better than the \textsf{Suggested Matching} algorithm, e.g., on the second-class edges.
By carefully leveraging the portions of the three stages, they show that their algorithm is at least $0.645$-competitive.

\section{Evolving Suggested Matching Algorithm}
	\label{sec:ESM}
	
In this section, we propose the \textsf{Evolving Suggested Matching} (ESM) algorithm that generalizes several existing algorithms for the online stochastic matching problem, including that of Yan~\cite{Y22}.
	The algorithm is equipped with a non-decreasing \emph{activation} function $f: [0,1] \to [0,2]$.
	We first present the algorithm in its general form and then provide a lower bound on the competitive ratio in terms of the function $f$.
	By carefully fixing the activation function (in the next section), we show that the competitive ratio is at least $0.650$.
	
	\subsection{The Algorithm} \label{ssec:algorithm}
	
	Our algorithm is inspired by the MSM algorithm from \cite{Y22}, in which second-class edges follow different matching strategies at different stages of the algorithm.
	The high-level idea of our algorithm is to introduce a non-decreasing \emph{activation} function $f: [0,1] \to [0,2]$ to make this transition happen smoothly.
	As in \cite{Y22}, we only consider the kernel instances in which each online vertex type is either first-class (having only one neighbor) or second-class (having exactly two neighbors).
	In the ESM algorithm, when an online vertex of type $i$ arrives at time $t\in [0,1]$,
	\begin{itemize}
		\item if $i$ is a first-class type, then it proposes to its unique neighbor $j$. That is, if $j$ is unmatched then we include the edge $(i,j)$ into the matching; otherwise $i$ is discarded.
		\item if $i$ a second-class type, then it chooses a neighbor $j_1$ uniformly at random as its \emph{first-choice}, and let the other neighbor $j_2$ be its \emph{second-choice}.
		Then with probability $\min\{ f(t),1 \}$, $i$ proposes to $j_1$.
		If the proposal is made and $j_1$ is unmatched then the edge $(i,j_1)$ is included in the matching and this round ends; if the proposal is made but $j_1$ is already matched, then $i$ proposes to $j_2$ with probability $\max\{ f(t)-1, 0 \}$.
	\end{itemize}
	
	The detailed description of the algorithm can be found in \Cref{algo:ESM}.
	
	\begin{algorithm}[th]
		\caption{Evolving Suggested Matching algorithm}
		\label{algo:ESM}
		\KwIn{A kernel instance with graph $G=(I\cup J, E)$, arrival rates $\{ \lambda_i \}_{i\in I}$, the optimal solution $\boldsymbol{x}$ to $\text{LP}_{\text{JL}}$, and an activation function $f$.}
		\KwOut{A matching $\+M$.}
		Initialize $\+M=\emptyset$ to be an empty matching\;
		\For{each online vertex of type $i$ arriving at time $t\in [0,1]$}
		{
			\eIf{$i$ is a first-class online vertex type}
			{
				\tcp{Propose to its unique first-class neighbor $j$}
				\If{$j$ is unmatched}
				{
					$\+M\gets \+M\cup \set{(i,j)}$\;
				}  
			}
			{ 
				choose a neighbor $j_1$ uniformly at random and let $j_2$ be the other neighbor\;
				$r_1,r_2\sim \!{Unif}[0,1]$\;
				\If{$ r_1\leq f(t) $}
				{
					\tcp{Propose to $j_1$} 
					\If {$j_1$ is unmatched}
					{
						$\+M\gets \+M\cup \set{(i,j_1)}$\;     
					}
					\ElseIf{$r_2\leq f(t)-1$}
					{
						\tcp{Propose to $j_2$}
						\If{$j_2$ is unmatched}
						{
							$\+M\gets \+M\cup \set{(i,j_2)}$\;
						}
					}
				}
			}
		}
		\Return{$\+M$.}
	\end{algorithm}
	
	We remark that the activation function $f$ controls how ``aggressive'' the second-class online vertex types propose to their neighbors.
	For example, when $f(t) = 0$, the second-class online vertex arriving at time $t$ will be discarded immediately without making any matching proposal; if $f(t) = 1$ then it will only propose to its first-choice; if $f(t) = 2$ then it will first propose to its first-choice and if the proposal is unsuccessful, then it will propose to its second-choice.
	Therefore, with different choices of the activation function, the ESM is general enough to capture many of the existing algorithms for the online stochastic matching problem, including the \textsf{Suggested Matching} algorithm (with $f(t) = 1$ for all $t\in[0,1]$); the \textsf{Two-Choice} algorithm (with $f(t) = 2$ for all $t\in[0,1]$) and that of Yan~\cite{Y22} (with $f(t) = 0$ when $t\leq 0.05$; $f(t)=1$ when $t\in (0.05, 0.75)$ and $f(t) = 2$ when $t\geq 0.75$).
	In the following, we derive a lower bound on the competitive ratio of the algorithm in terms of the activation function $f$.
	A specific choice of $f$ will be decided in the next section to optimize the lower bound on the ratio.
	
	\subsection{Extended Online Vertex Types}
	
	For convenience of analysis, in the following, we make use of the properties of the Poisson process and present an equivalent description of the ESM algorithm.
	Specifically, upon the arrival of a second-class online vertex of type $i$ at time $t$, suppose that $j_1$ is chosen as the first-choice and $j_2$ is the second-choice.
	\begin{itemize}
		\item If $r_1 > f(t)$ then we call the online vertex of type $i(\bot, \bot)$, indicating that it will not propose to any of its two choices;
		\item If $r_1 \leq f(t)$ and $r_2 \geq f(t)-1$ then we call the online vertex of type $i(j_1, \bot)$, indicating that it will only propose to its first-choice;
		\item If $r_1 \leq f(t)$ and $r_2 < f(t)-1$ then we call the online vertex of type $i(j_1, j_2)$, indicating that it will propose to its both choices unless it gets matched. 
	\end{itemize}
	
	Recall that type $i$ arrives following a Poisson process with rate $\lambda_i$.
	Hence the aforementioned types arrive following Poisson processes with time-dependent arrival rates (depending on the activation function $f$).
	Since the first-choice is chosen uniformly at random between the two neighbors, and $r_1,r_2$ are uniformly distributed in $[0,1]$, we can characterize the arrival rates of each specific extended vertex type as follows.
	
	\begin{proposition}\label{prop:extendedpoisson}
		For any second-class online vertex type $i$ with neighbors $\set{j_1,j_2}$, at time $t\in [0,1]$
		\begin{itemize}
			\item the arrival rate of type $i(j_1,\bot)$ and $i(j_2,\bot)$ are both $\frac{\lambda_i}{2}\cdot \min\{ f(t),1 \}\cdot \min\{ 2 - f(t), 1 \}$;
			\item the arrival rate of type $i(j_1, j_2)$ and $i(j_2, j_1)$ are both $\frac{\lambda_i}{2}\cdot \min\{ f(t),1 \}\cdot \max\{ f(t) - 1, 0 \}$.
		\end{itemize}
		
		Moreover, the Poisson processes describing the arrivals of online vertex of type $i(j,j')$, for all $i\in I$ and $j,j'\in J\cup \{\bot\}$ are independent. 
	\end{proposition}
	
	In the following, we use $i(j, j')$ to describe an extended type, where $i$ is a second-class online vertex and each of $j,j'$ is either $\bot$ or a neighbor of $i$.
	Under this independent Poisson process modeling on the arrivals of the extended types, we give an equivalent description of the ESM algorithm in Algorithm~\ref{algo:ESMextended}.
	
	\begin{algorithm}[th]
		\caption{ESM algorithm with extend online vertex types}
		\label{algo:ESMextended}
		\KwIn{A kernel instance with extended vertex types.}
		\KwOut{A matching $\+M$.}
		Initialize $\+M=\emptyset$ to be an empty matching\;
		\For{each online vertex of type $i$}
		{
			\eIf{$i$ is a first-class online vertex type}
			{
				%				\tcp{Propose to its unique first-class neighbor $j$}
				\If{its neighbor $j$ is unmatched}
				{
					$\+M\gets \+M\cup \set{(i,j)}$\;
				}  
			}
			{
				Suppose the type is $i(j,j')$\;
				\If{$j \neq \bot$}
				{
					%					\tcp{Propose to $j_1$}
					{
						\If {$j$ is unmatched}
						{
							$\+M\gets \+M\cup \set{(i,j)}$\;     
						}
						\ElseIf{$j'\neq \bot$}
						{
							%							\tcp{Propose to $j_2$}
							\If {$j'$ is unmatched}
							{
								$\+M\gets \+M\cup \set{(i,j')}$\;     
							}
						}
					}
				}
			}           
		}
		\Return{$\+M$.}
	\end{algorithm}
	
	In the remaining analysis, we say that an online vertex is of type $i(j, *)$ if its extended type is $i(j, j')$ for some $j' \in J\cup \{ \bot \}$.
	Likewise we define $i(*,*)$, $*(j, *)$, $*(*, j)$, etc. These notations only consider the second class arrivals.
	Note that for all second-class online vertex type $i\in I$ and any of its neighbor $j$, the arrival rate of type $i(j, *)$ at time $t$ is $\frac{\lambda_i}{2}\cdot \min\{ f(t),1 \}$ and the arrival rate of type $i(*,*)$ at any time is always $\lambda_i$.
 
 % Similarly, for any offline vertex $j$, the arrival rate of type $*(j,*)$ and $*(*,j)$ are both ${1-y_j}$ which corresponds to the second class arrival among its neighbors.

\subsection{Matching Probability of Edges}
	
	By Lemma~\ref{lemma:analysis-framework}, to derive a lower bound on the competitive ratio of the algorithm, it suffices to lower bound the probability that an arbitrarily fixed edge $(i,j) \in E$ is matched by the ESM algorithm.
	Let $M_{i j} \in \{ 0, 1\}$ be the indicator of whether $(i,j)$ is matched by the algorithm, and $F(t) = \int_{0}^t f(x) \, dx$.
 	Let $U_j(t)\in \{ 0,1 \}$ be the indicator of whether the offline vertex $j$ is unmatched at time $t$, i.e., $U_j(t) = 1$ if and only if $j$ is unmatched at time $t$.
	In the following, we provide a lower bound on $\Pr{M_{i j} = 1}$.
	
	\paragraph{First-Class Edge.}
	We first consider the case when $(i,j)$ is a first-class edge.
	By the design of the algorithm, the edge will be included in the matching by the algorithm if an online vertex of type $i$ arrives, and $j$ is unmatched, because $i$ will always propose to $j$ upon its arrival.
	% Let $U_j(t)\in \{ 0,1 \}$ be the indicator of whether the offline vertex $j$ is unmatched at time $t$, i.e., $U_j(t) = 1$ if and only if $j$ is unmatched at time $t$.
	Since the arrival rate of type $i$ is $\lambda_i$, we immediately have the following.
	
	\begin{lemma} \label{lemma:matching-prob-first-class}
		For any first-class edge $(i,j)\in E$, we have
		\begin{equation*}
			\Pr{M_{ij}=1}=  \int_{0}^{1} \lambda_{i} \cdot  \Pr{U_j(t)=1} \, dt. 
		\end{equation*}
	\end{lemma}
	
	\paragraph{Second-Class Edge.}
	Now suppose that $(i,j)\in E$ is a second-class edge, and let $j'$ be the other neighbor of $i$.
	There are two events that will cause edge $(i,j)\in E$ being matched by the algorithm:
	\begin{itemize}
		\item an online vertex of type $i(j, *)$ arrives, and $j$ is unmatched;
		\item an online vertex of type $i(j' ,j)$ arrives, $j'$ is already matched and $j$ is unmatched.
	\end{itemize}
	Let $t^*:=\sup \set{t\in [0,1]: f(t)\leq 1}$. We observe that
	\begin{itemize}
		\item at time $t \leq t^*$, the arrival rate of $i(j, *)$ is $\frac{\lambda_i}{2}\cdot f(t)$; after time $t^*$, the arrival rate is $\frac{\lambda_i}{2}$;
		\item before time $t^*$, the arrival rate of $i(j', j)$ is $0$; at time $t > t^*$, the arrival rate is $\frac{\lambda_i}{2}\cdot (f(t)-1)$.
	\end{itemize}
	Therefore, we have the following characterization on $\Pr{M_{i j}=1}$.
	\begin{align*}
		\Pr{M_{ij}=1} &= \int_{0}^{t^*} \frac{\lambda_{i}}{2}\cdot  f(t)\cdot \Pr{ U_{j}(t)=1}\, dt +  \int_{t^*}^{1}  \frac{\lambda_{i}}{2}\cdot \Pr{U_{j}(t)=1} \, dt  \\
		& \quad + \int_{t^*}^{1} \frac{\lambda_{i}}{2} \cdot (f(t)-1) \cdot \Pr{U_{j}(t)=1, U_{j'}(t)=0} \, dt \\
		&= \int_{0}^{1} \frac{\lambda_{i}}{2}\cdot  f(t)\cdot \Pr{ U_{j}(t)=1}\, dt 
		- \int_{t^*}^{1} \frac{\lambda_{i}}{2} \cdot (f(t)-1) \cdot \Pr{U_{j}(t)=1, U_{j'}(t)=1} \, dt.
	\end{align*}
	
	By upper bounding $\Pr{U_{j}(t)=1, U_{j'}(t)=1}$, we derive the following.
	
	\begin{lemma} \label{lemma:matching-prob-second-class}
		For any second-class edge $(i,j)\in E$, we have
		\begin{align*}
			\Pr{M_{ij}=1}& \geq \int_{0}^{1} \frac{\lambda_{i}}{2}\cdot  f(t)\cdot \Pr{ U_{j}(t)=1}\, dt - \int_{t^*}^{1} \frac{\lambda_{i}}{2} \cdot (f(t)-1) \cdot e^{-y_{j}\cdot t^*-(2-y_{j})\cdot F(t^*)-2(t-t^*)} \, dt.
		\end{align*}
	\end{lemma}
	\begin{proof}
		To prove the lemma, it suffices to argue that for any $t > t^*$, we have
		\begin{equation*}
			\Pr{U_{j}(t)=1,U_{j'}(t)=1}\leq e^{-y_{j}\cdot t^* -(2-y_{j})\cdot F(t^*)-2(t-t^*)}.
		\end{equation*}
		
		Observe that if any online vertex of type $i\in N_1(j)$ or $*(j,*)$ arrives before time $t$, then $j$ will be matched before time $t$. The same holds for the offline vertex $j'$.
		The arrival rate of type $N_1(j)\cup*(j,*)$ is
		\begin{itemize}
			\item at time $x \leq t^*$: $\sum_{i\in N_1(j)} \lambda_i + \sum_{i\in N_2(j)} \frac{\lambda_i}{2} \cdot f(x) = y_j + (1-y_j)\cdot f(x)$.
			\item after time $t^*$: $\sum_{i\in N_1(j)} \lambda_i + \sum_{i\in N_2(j)} \frac{\lambda_i}{2} = y_j + (1-y_j) = 1$.
		\end{itemize}
		Similarly the arrival rate of type $N_1(j')\cup*(j',*)$ is $ y_{j'} + (1-y_{j'})\cdot f(x)$ at time $x \leq t^*$ and $1$ after time $t^*$.
		Since $\Pr{U_{j}(t)=1,U_{j'}(t)=1}$ is at most the probability that no online vertex of type $N_1(j)\cup*(j,*)$ or $N_1(j')\cup*(j',*)$ arrives before time $t$, and the arrivals of types $N_1(j)\cup*(j,*)$ or $N_1(j')\cup*(j',*)$ are independent, we have
		\begin{align*}
	    \Pr{U_{j}(t)=1,U_{j'}(t)=1}& \leq e^{- \int_0^{t^*} \left(y_j + (1-y_j)\cdot f(x) \right) \, dx - \int_{t^*}^t 1 \, dx } \cdot e^{-\int_0^{t^*} \left( y_{j'} + (1-y_{j'})\cdot f(x) \right) \, dx - \int_{t^*}^t 1 \, dx } \\
			& = e^{-(y_{j} + y_{j'})\cdot t^* -(2-y_{j}-y_{j'})\cdot F(t^*)-2(t-t^*)} \\
			& \leq e^{-y_{j} \cdot t^* -(2-y_{j})\cdot F(t^*)-2(t-t^*)},
		\end{align*}
		where the last inequality holds because $F(t^*) \leq t^*$ (by definition of $t^*$) and $y_{j'}\geq 0$.
	\end{proof}
	
	Given Lemma~\ref{lemma:matching-prob-first-class} and~\ref{lemma:matching-prob-second-class}, to provide a lower bound on $\Pr{M_{i j} = 1}$, it remains to lower bound $\Pr{ U_{j}(t)=1}$ (in terms of the activation function $f$).

\subsection{Lower Bounding \texorpdfstring{$\Pr{ U_{j}(t)=1}$}{}}
	
	In this section, we fix an arbitrary offline vertex $j$ and provide a lower bound on $\Pr{ U_{j}(t)=1}$.
	To begin with (and as a warm-up), we first establish a loose lower bound as follows.
	
	\begin{lemma}\label{lemma:lose-lower-bound}
		For any offline vertex $j\in J$ and any time $t\in [0,1]$, we have
		\begin{equation*}
			\Pr{U_j(t)=1}\geq e^{-y_j t -(1-y_j)F(t)}.
		\end{equation*}
		Moreover, the above equation holds with equality when $t\leq t^*$.
	\end{lemma}
	\begin{proof}
		As before, we characterize the events that will cause $j$ being matched.
		Observe that if $j$ is matched before time $t$, then at least one of the following must happen before time $t$:
		\begin{itemize}
			\item a first-class online vertex type $i\in N_1(j)$ arrives;
			\item an online vertex of type $*(j,*)$ arrives;
			\item an online vertex of type $*(*, j)$ arrives.
		\end{itemize}
		
		Note that the third event will only happen after time $t^*$.
        Moreover, it will contribute to the matching of vertex $j$ only if the first-choice of the online vertex is matched upon its arrival. 
		However, for the purpose of lower bounding $\Pr{U_j(t)=1}$, we only look at the arrivals without caring whether a proposal to $j$ is made.
		Since the combined arrival rate of types $*(j,*)$ and $*(*,j)$ at time $x\in [0,1]$ is $(1-y_j)\cdot f(x)$, and the total arrival rate of vertices in $N_1(j)$ is $y_j$, the lemma follows immediately.
	\end{proof}
	
	It is apparent that the above lower bound on $\Pr{U_j(t)=1}$ can be improved when $t > t^*$ because, in the above analysis, we use the event that ``an online vertex of type $*(*, j)$ arrives'' to substitute that ``an online vertex of type $*(*, j)$ arrives and its proposal to its first-choice fails''.
	The advantage of this substitution is that now the events that may contribute to $j$ being matched are independent and it is convenient in lower bounding the probability that none of them happens.
	On the other hand, it is reasonable to believe that via lower bounding the probability that a vertex of type $*(*, j)$ fails in matching its first-choice, a better lower bound on $\Pr{U_j(t)=1}$ can be derived.
	However, this requires a much more careful characterization of these events, because some of them are not independent.
	Specifically, suppose that $j$ is not matched at time $x$, and an online vertex of type $i(j_1, j)$ arrives.
	The online vertex will contribute to $j$ being matched only if $j_1$ is matched.
	Therefore the contributions of types $i(j_1, j)$ and $i'(j_1, j)$ are not independent random events because they both depend on the matching status of $j_1$.
	Moreover, the contributions of types $i(j_1, j)$ and $i'(j_2, j)$ may also be dependent because whether $j_1$ and $j_2$ are matched might not be independent, e.g., they might have common neighbors.
	
	\smallskip
	
	Therefore, in order to provide a better lower bound on $\Pr{U_j(t)=1}$, it is inevitable to take into account the dependence on the random events.
	To enable the analysis, we introduce the following useful notations.
	
	\begin{definition} [Competitor]
		We call an offline vertex $j'$ a \emph{competitor} of $j$ if $N_2(j) \cap N_2(j') \neq \emptyset$.
		We use $\mathcal{C}(j) = \{ j_1, j_2,\ldots,j_K \}$ to denote the set of competitors of $j$.
		For each $j_k \in \mathcal{C}(j)$, we use $c_k = \sum_{i\in N_2(j)\cap N_2(j_k)} \frac{\lambda_i}{2}$ to denote the total arrival rate of types $\{ i(j_k, *) \}_{i\in N_2(j)\cap N_2(j_k)}$.
	\end{definition}
	
	Note that by definition we have $\sum_{k = 1}^K c_k = 1 - y_j$.
	In the following, we show the following improved version of Lemma~\ref{lemma:lose-lower-bound} when $t > t^*$.
	
	\begin{lemma} \label{lemma:improved-lower-bound}
		For any offline vertex $j\in J$ and any time $t\in [t^*,1]$, we have
		\begin{equation*}
			\Pr{U_j(t)=1}\geq e^{-y_{j}\cdot t -(1-y_{j})\left( e^{-F(1)} F(t^*)+ (1-e^{-F(1)}) F(t)+e^{-F(1)} (t-t^*)\right) }.
		\end{equation*}when $F(1)\geq 1$.
	\end{lemma}
	\begin{proof}
		By Lemma~\ref{lemma:lose-lower-bound}, we have $\Pr{U_j(t^*)=1} = e^{-y_{j} \cdot t^* - (1-y_{j}) F(t^*)}$.
		Hence the statement is true when $t = t^*$.
		Observe that $\Pr{U_j(t)=1} = \Pr{U_{j}(t) = 1 \mid U_{j}(t^*)=1}\cdot \Pr{U_j(t^*)=1}$.
		To prove the lemma, it suffices to show that for all $t > t^*$, we have
		\begin{equation}
			\Pr{U_{j}(t) = 1 \mid U_{j}(t^*)=1}\geq e^{ - y_j(t-t^*) -(1-y_{j})\left( (1-e^{-F(1)}) (F(t) - F(t^*))+e^{-F(1)} (t-t^*)\right) }. \label{eqn:j-unmatched-conditioned}
		\end{equation}
		
		For ease of notation, we use $h(t)$ to denote the LHS of the above.
		Note that $h : [t^*,1] \to [0,1]$ is a decreasing function with $h(t^*) = 1$.
		Fix any $t > t^*$.
		Conditioned on $j$ being unmatched at time $t^*$, $j$ will be matched at or before time $t$ if any of the following events happen during the time interval $[t^*, t]$:
		\begin{itemize}
			\item a first-class online vertex $i \in N_1(j)$ arrives;
			\item an online vertex of type $*(j,*)$ arrives;
			\item an online vertex of type $*(*,j)$ arrives, and it fails matching its first-choice.
		\end{itemize}
		
		Since the total arrival rate of the first two events is $1$ and the arrival rate of $*(j_k, j)$ at time $x > t^*$ is $c_k\cdot (f(x)-1)$, we have the following
		\begin{align*}
			& \quad \Pr{U_{j}(t)=0\mid U_{j}(t^*)=1} \\
    &= \int_{t^*}^t  \bigg(\Pr{U_{j}(x)=1\mid U_{j}(t^*)=1} \notag  + (f(x)-1)\cdot \sum_{k=1}^K c_k \cdot \Pr{U_{j_k}(x)=0,U_{j}(x)=1\mid U_{j}(t^*)=1} \bigg) \, dx.
		\end{align*}
		
		Observe that $\Pr{U_{j}(t)=0\mid U_{j}(t^*)=1} = 1 - h(t)$.
		The above equation implies that
		\begin{align*}
			1 - h(t) = & \int_{t^*}^t  \left( 1  + (f(x)-1)\cdot \sum_{k=1}^K c_k \cdot \Pr{U_{j_k}(x)=0\mid U_{j}(x)=1} \right)\cdot h(x) \, dx.
		\end{align*}
		
		Solving the above using the standard differential equation, we have
		\begin{equation} \label{eqn:expression-h}
			h(t) = e^{-\int_{t^*}^t  g(x) \, dx},
		\end{equation}
		where $g(x) = 1  + (f(x)-1)\cdot \sum_{k=1}^K c_k \cdot \Pr{U_{j_k}(x)=0\mid U_{j}(x)=1}$.
        To provide an upper bound on $g(x)$, we first establish the following useful claim.
  
		\begin{claim} \label{claim:at-least-1-e^{-F(1)}}
			For all $x\geq t^*$ and $j'\in J\setminus \set{j}$, we have $\Pr{U_{j'}(x)=0\mid U_j(x)=1} \leq 1 - e^{-F(1)}$ when $F(1)>1$.
		\end{claim}
		\begin{proof}
			We prove the equivalent statement that $\Pr{U_{j'}(x)=1 \mid U_j(x)=1} \geq e^{-F(1)}$.
			As before, we first list the events that may cause $j'$ being matched:
			\begin{itemize}
				\item a first-class online vertex $i \in N_1(j')$ arrives;
				\item an online vertex of type $*(j',*)$ arrives;
				\item an online vertex of type $*(*,j')$ arrives.
			\end{itemize}
			
			We call the above types the \emph{key} types and use $A(x)$ to denote the event that none of the key types arrive before time $x$.
			Note that if $A(x)$ happens then $j'$ is guaranteed to be unmatched at time $x$.
			Hence we have
			\begin{align*}
				\Pr{U_{j'}(x)=1\mid U_j(x)=1} =& \frac{\Pr{U_{j'}(x)=1, U_j(x)=1}}{\Pr{U_j(x)=1}} \\
				\geq &\frac{ \Pr{A(x), U_j(x)=1}}{\Pr{U_j(x)=1}} = \frac{\Pr{A(x)}\cdot \Pr{U_j(x)=1\mid A(x)}}{\Pr{U_j(x)=1}}.
			\end{align*}
		
			Since the key types related to event $A(x)$ arrive independently, we have
			\begin{equation*}
				\Pr{A(x)} = e^{-y_j'\cdot x -(1-y_j')\cdot F(x)} \geq e^{-y_j' - (1-y_j')\cdot  F(1)} \geq e^{-F(1)},
			\end{equation*}
			where the last equality holds from the assumption that $F(1)\geq 1$.
			
			Given the above, it remains to show that
			\begin{equation}\label{eqn:monotone}
				\Pr{U_j(x)=1\mid A(x)} \geq \Pr{U_j(x)=1}.
			\end{equation}
    The statement can be proved by the coupling argument. Given the same set of randomness, let ${\+M}_1$ and $\+{M}_2$ be the matchings obtained by the ESM algorithm and the one neglecting all arrivals of key types before time $x$ respectively. We next show that $j\notin \+M_2$ if $j\notin \+M_1$ for any specified randomness, which implies \cref{eqn:monotone} immediately. Note that $j\notin \+M_1$ means 
    \begin{enumerate}
        \item  no online vertex of type $i\in N_1(j)$, $*(j,*)$ and $*(*,j)$ arrive, or
        \item  some online vertex of type $i(j'',j)$ arrive but $j''$ is not matched.
    \end{enumerate}
    In the first circumstance, it is obvious to have $j\notin \+M_2$; Otherwise, it suffices to show that $j''\notin \+M_2$ when $j''\notin \+M_1$ before the time $x'$ that the online vertex of type $i(j'',j)$ arrive. Again, we can repeat a similar argument till the beginning where all vertices are unmatched. Hence the inequality holds and our claim follows immediately.
    \end{proof}
   % Recall that conditioned on $A(x)$, no vertex of a key type arrives before time $x$.

			% Intuitively speaking, conditioned on these online vertex types not arriving, the probability of the offline vertex $j$ being unmatched can only be higher.
			% Formally speaking, imagine running the ESM algorithm in two parallel worlds, with the same set of randomness, up to time $x$.
			% Suppose in one of these two worlds we remove all arrivals of key types (if any).
			% Then it is clear that if $j$ is unmatched in the world with the key types then it will also be unmatched in the world without.
		
		Given the above claim, we have
		\begin{equation*}
			g(x) \leq 1  + (f(x)-1)\cdot \sum_{k=1}^K c_k \cdot \left( 1 - e^{-F(1)} \right) 
			= y_j + (1-y_j)\cdot \left( \left( 1-e^{-F(1)}\right) \cdot  f(x) + e^{-F(1)} \right).
		\end{equation*}
		
		Plugging the above upper bound on $g(x)$ into Equation~\eqref{eqn:expression-h}, we obtain Equation~\eqref{eqn:j-unmatched-conditioned}, and thus finish the proof of the lemma.
	\end{proof}
 
\subsection{Putting Things Together}\label{sec:final-lower-bounds}
	
	Plugging the lower bounds on $\Pr{U_j(t) = 1}$ we have proved in Lemma~\ref{lemma:lose-lower-bound} and~\ref{lemma:improved-lower-bound} into Lemma~\ref{lemma:matching-prob-first-class} and~\ref{lemma:matching-prob-second-class}, we obtain the following when the activation function satisfying $F(1)\geq 1$.	
	For any first-class edge $(i,j) \in E$, 
	\begin{align}
		 \frac{\Pr{M_{ij}=1}}{\lambda_{i} } &\geq     \int_{0}^{t^*}    e^{-y_j\cdot t -(1-y_j)F(t)} \,dt \nonumber  \\
		&\quad + \int_{t^*}^1  e^{-y_{j}\cdot t -(1-y_{j}) \left(e^{-F(1)} F(t^*)+ (1-e^{-F(1)}) F(t)+e^{-F(1)} (t-t^*)\right)}\, dt. \label{eqn:final-lower-bound-1}
	\end{align}
	
	For any second-class edge $(i, j)\in E$,
	\begin{align}
		\frac{\Pr{M_{ij}=1} }{\lambda_i/2}
		& \geq  \int_{0}^{t^*}    f(t)\cdot e^{-y_j\cdot t -(1-y_j)F(t)} \,dt \nonumber \\
		& \quad + \int_{t^*}^1 f(t)\cdot  e^{-y_{j}\cdot t-(1-y_{j}) \left( e^{-F(1)} F(t^*)+ (1-e^{-F(1)}) F(t)+e^{-F(1)} (t-t^*) \right)} \,dt \nonumber \\
		& \quad - \int_{t^*}^1 \left( f(t)-1\right) \cdot  e^{-y_{j}\cdot t^* -(2-y_{j})\cdot F(t^*)-2\cdot(t-t^*)}  \,dt . \label{eqn:final-lower-bound-2}
	\end{align}
	
	By Lemma~\ref{lemma:analysis-framework}, to show that the ESM algorithm is $\alpha$-competitive, it remains to design an appropriate activation function $f$ such that for all $y_j \leq 1-\ln 2$, the RHS of Equations~\eqref{eqn:final-lower-bound-1} and~\eqref{eqn:final-lower-bound-2} are both at least $\alpha$ and $F(1)\geq 1$. 
%  The details of the construction of the activation function are left to the appendix. Specifically, we show that
% \begin{lemma}\label{lemma:activate function design}
%     There exists a non-decreasing activation function $f$ with $F(1)\geq 1$ such that 
%     \begin{itemize}
%         \item for any first-class edge $(i,j)\in E$, $\frac{\Pr{M_{ij}=1}}{\lambda_{i} }\geq 0.650$;
%         \item for any second-class edge $(i,j)\in E$, $\frac{\Pr{M_{ij}=1} }{\lambda_i/2}\geq 0.650$.
%     \end{itemize}
% \end{lemma}
% Together with Lemma~\ref{lemma:analysis-framework} and the lower bounds we derived, we prove Theorem~\ref{theo:competitive-ratio}.

\section{Design of the Activation Function}

Following the previous analysis, in this section we compute an appropriate activation function $f$ with which the lower bounds we derive in \Cref{sec:ESM} are at least $0.650$.
For ease of notation, we let
\begin{equation*}
    z(t)= e^{-F(1)}\cdot F(t^*)+ \left( 1-e^{-F(1)} \right)\cdot F(t)+e^{-F(1)}\cdot (t-t^*),
\end{equation*}
and define two functions $r_1, r_2$ as follows.
For all $y \in [0,1-\ln 2]$, let
\begin{equation*}
    r_1(y)=\int_{0}^{t^*}  e^{-y  t -(1-y)F(t)} \,dt + \int_{t^*}^1  e^{-y t -(1-y) z(t)} \,dt,
\end{equation*}
\begin{align*}
    r_2(y)=&   \int_{0}^{t^*}    f(t) e^{-y t -(1-y)F(t)} \,dt
    + \int_{t^*}^1 f(t)  e^{-y t -(1-y) z(t)} \,dt\\
   &- \int_{t^*}^1 ( f(t)-1)  e^{-y  t^* -(2-y) F(t^*)-2(t-t^*)}  \,dt.
\end{align*}

Note that the RHS of Equations~\eqref{eqn:final-lower-bound-1} and~\eqref{eqn:final-lower-bound-2} are precisely $r_1(y_j)$ and $r_2(y_j)$, respectively.
Let $y^* = 1 - \ln 2$.
Recall that $y_j \leq y^*$ follows from the feasibility of solution $\boldsymbol{x}$.
Hence our goal is to find a non-decreasing function $f: [0,1] \to [0,2]$ such that $F(1)\geq 1$ (which is required in \Cref{lemma:improved-lower-bound}) and
\begin{equation*}
     \min_{y \in [0, y^*]} \set{r_1(y)} \geq 0.650, \quad \text{and} \quad
     \min_{y \in [0, y^*]} \set{r_2(y)} \geq 0.650.
\end{equation*}

To make this optimization problem tractable, our first idea is to restrict our choice of the function $f$ such that $r_1(y)$ and $r_2(y)$ achieve their minimum value at $y = y^*$.
To achieve this goal, we put constraints on $f$ so that the derivatives of both $r_1$ and $r_2$ on $y\in [0,y^*]$ are always non-positive (see \Cref{ssec: derivative constraints} for the details).
As a consequence, our goal now becomes finding a non-decreasing function $f$ meeting all constraints, such that the objective $\min \left\{ r_1(y^*), r_2(y^*) \right\}$ is as large as possible.
However, given the complex formulas on the objective and the constraints, it is hard to compute the function $f$ that optimizes the objective.
Our second idea is to bound this complicated continuous optimization problem by discretization. 
Specifically, we restrict our choice of the function $f$ to piecewise constant functions, which allows us to translate the continuous optimization problem into one with a fixed number of parameters (see \Cref{ssec: discretization}).
Finally, it remains to use a computer program to enumerate the choices of parameters that maximize the (discretized) objective.

\subsection{Derivative Constraints}\label{ssec: derivative constraints}

Take derivative of $r_1(y)$, we have
\begin{align*}
    \frac{d r_1}{d y} & = \int_{0}^{t^*} \left( -t+F(t) \right) e^{-y t -(1-y)F(t)} \,dt  + \int_{t^*}^1 \left( -t+z(t) \right) e^{-y t -(1-y) z(t)} \,dt\\
    &\leq \int_{0}^{t^*}  \left( F(t) - t \right) e^{-y^* t -(1-y^*)F(t)} \,dt  + \int_{t^*}^1 \left( z(t)-t \right) e^{-y^* t -(1-y^*) z(t)} \,dt,
\end{align*}
where the inequality follows from the fact that
\begin{align*}
    \frac{d^2 r_1}{d y^2}=\int_0^{t^*} \left( F(t)-t \right)^2 e^{-y t -(1-y)F(t)} \,dt + \int_{t^*}^1 \left( z(t)-t \right)^2 e^{-y t -(1-y)z(t)} \,dt \geq 0.
\end{align*}

Hence to ensure that $r_1(y)$ admits its minimum value at $y^*$, it suffices to have
\begin{align}\label{eqn:cons1}
 \int_{0}^{t^*} \left(  F(t)-t \right) e^{-y^* t -(1-y^*)F(t)} \,dt  + \int_{t^*}^1 \left( z(t)-t \right) e^{-y^* t -(1-y^*) z(t)} \,dt\leq 0.
\end{align}

% \subsubsection{Second-Class Edge}
\medskip

Now we consider $r_2(y)$.
Take derivative of $r_2(y)$, we have
\begin{align}
     \frac{dr_2}{dy} & = \int_{0}^{t^*} \left( -t+F(t)\right)f(t) e^{-y t -(1-y)F(t)} \,dt  + \int_{t^*}^1\left(-t+z(t)\right)f(t) e^{-y t -(1-y) z(t)} \,dt \notag \\
     & \quad -\int_{t_*}^{1} \left( f(t)-1 \right) \left( -t^*+F(t^*)\right) e^{-y t^* -(2-y) F(t^*)-2(t-t^*)}  \,dt \notag \\
     & \leq \int_{0}^{t^*} \left(  F(t)-t \right)f(t) e^{-y t -(1-y)F(t)} \,dt  + \int_{t^*}^1 \left( z(t)-t \right)f(t) e^{-y t -(1-y) z(t)} \,dt \notag \\
     & \quad + \int_{t^*}^{1}  \left( f(t)-1 \right) \left( t^*-F(t^*)\right) e^{-2 F(t^*)-2 (t-t^*)} \,dt \label{eqn:immediateterm}\\
     & \leq \int_{0}^{t^*}   \left( F(t)-t \right) f(t) e^{-y^* t -(1-y^*)F(t)} \,dt  + \int_{t^*}^1 \left( z(t)-t \right)f(t) e^{-y^* t -(1-y^*) z(t)} \,dt \notag \\
     & \quad + \int_{t_*}^{1}  \left( f(t)-1 \right) \left( t^*-F(t^*)\right)  e^{-2 F(t^*)-2 (t-t^*)}  \,dt \notag,
\end{align}
where the first inequality follows from $F(t^*)\leq t^*$ (by definition of $t^*$) and the last inequality holds since the derivative of~\eqref{eqn:immediateterm} with respect to the variable $y$ is non-negative, as can be shown as follows:
\begin{align*}
   &\int_0^{t^*}\left( F(t)-t \right)^2 f(t) e^{-y t -(1-y)F(t)} \,dt + \int_{t^*}^1\left( z(t)-t \right)^2 f(t) e^{-y t -(1-y)z(t)} \,dt \geq 0,
\end{align*} which implies that the maximum value of~\eqref{eqn:immediateterm} is achieved when $y=y^*$

Hence to ensure that $r_2(y)$ admits its minimum value at $y^*$, it suffices to have
\begin{align}\label{eqn:cons2}
   &\int_{0}^{t^*}   \left( F(t)-t \right)f(t) e^{-y^* t -(1-y^*)F(t)} \,dt  + \int_{t^*}^1 \left( z(t)-t \right)f(t) e^{-y^* t -(1-y^*) z(t)}\,dt  \notag \\
     & + \int_{t_*}^{1}  \left( f(t)-1 \right) \left( t^*-F(t^*)\right)  e^{-2 F(t^*)-2 (t-t^*)}  \,dt \leq 0.
\end{align}

In summary, it remains to find a non-decreasing function $f: [0,1] \to [0,2]$ that satisfies constraints~\eqref{eqn:cons1} and~\eqref{eqn:cons2}, and that $F(1)\geq 1$, so that $\min\{ 
r_1(y^*), r_2(y^*) \}$ is as large as possible.

\subsection{Discretization}\label{ssec: discretization}

In this section we restrict our choice of activation function to non-decreasing piecewise constant functions.
Let $m\geq 1$ be an integer that controls the granularity of discretization.
We discretize the domain $(0,1]$ of function $f$ into $m$ intervals with equal width, on each of which the function value is a fixed constant.
Specifically, the $i$-th interval is $( \frac{i-1}{m},\frac{i}{m} ]$, where $i\in [m] := \{1,2,\ldots,m\}$.
For any non-decreasing piecewise constant function $f$ with $m$ intervals, we will use the following succinct representation:
\begin{itemize}
    % \item Let $0=t_1<t_2<\dots<t_{m+1}=1$ be the breakpoints of the partition in $[0,1]$;
    \item Let $0\leq f_1 < f_2 < \dots < f_m \leq 2$.
    For all $i\in [m]$ and $t\in ( \frac{i-1}{m},\frac{i}{m} ]$, we have $f(t) = f_i$.
\end{itemize}

% Accordingly, we use $\{F_s\}_{s\in [m]}$ to denote the accumulated value of the activation function up to $t_s$, i.e., $F_s = \frac{1}{m}\cdot \sum_{i\in [s]} f_{i}$.
We use $k = \arg\max\{ i : f_i\leq 1 \}$ to denote the maximum index with which the activation function is at most $1$.
Therefore we have $t^* = \frac{k}{m}$.
It remains to replace the function $f$ with its succinct representation $(f_1,f_2,\ldots,f_m)$ in constraints~\eqref{eqn:cons1},~\eqref{eqn:cons2} and that $F(1)\geq 1$, so that the LHS of all these constraints become a function of $f_1,f_2,\ldots,f_m$.
For example, $F(1)\geq 1$ translates into
\begin{equation*}
    f_1 + f_2 + \ldots + f_m \geq m.
\end{equation*}

The translations for the constraints~\eqref{eqn:cons1} and~\eqref{eqn:cons2} into their discretized versions are straightforward but tedious.
In the following we pick the first term
\begin{equation*}
    \int_{0}^{t^*} \left(  F(t)-t \right)\cdot e^{-y^* t -(1-y^*)F(t)} \,dt
\end{equation*}
of constraint~\eqref{eqn:cons1} as a demonstration and omit the details for the translation of other parts.

Using the definitions we have introduced, we obtain
\begin{align*}
    & \int_{0}^{t^*} \Bigg(  F(t)-t \Bigg)\cdot e^{-y^* t -(1-y^*)F(t)} \,dt 
    = \sum_{i=1}^k \int_{\frac{i-1}{m}}^{\frac{i}{m}} \Bigg(  F(t)-t \Bigg)\cdot e^{-y^* t -(1-y^*)F(t)} \,dt \\
    \leq &\ \sum_{i=1}^k \Bigg( \Bigg(  \frac{1}{m}\cdot \sum_{s\leq i-1} f_s - \frac{i-1}{m} \Bigg)\cdot \int_{\frac{i-1}{m}}^{\frac{i}{m}} e^{-y^* t -(1-y^*)F(t)} \,dt \Bigg) \\
    = &\ \sum_{i=1}^k \Bigg( \Bigg(  \frac{1}{m}\cdot \sum_{s\leq i-1} f_s - \frac{i-1}{m} \Bigg)\cdot \int_{\frac{i-1}{m}}^{\frac{i}{m}} e^{-y^* t -(1-y^*)\Bigg( \frac{1}{m}\cdot \sum_{s\leq i-1} f_s + f_i\cdot (t-\frac{i-1}{m}) \Bigg)} \,dt \Bigg) \\
    % = &\ \sum_{i=1}^k \left( \left( \frac{1}{m}\cdot \sum_{s\leq i-1} f_s - \frac{i-1}{m} \right) \cdot e^{ (1-y^*)\left( \frac{1}{m}\cdot \sum_{s\leq i-1} (f_i - f_s) \right)} \cdot \int_{\frac{i-1}{m}}^{\frac{i}{m}} e^{- (y^* + (1-y^*) f_i) t} \,dt \right)\\
    = &\ \sum_{i=1}^k \Bigg( \Bigg( \frac{1}{m}\cdot \sum_{s\leq i-1} f_s - \frac{i-1}{m} \Bigg) \cdot e^{ (1-y^*)\left( \frac{1}{m}\cdot \sum_{s\leq i-1} (f_i - f_s) \right)} \cdot \frac{e^{- (y^* + (1-y^*) f_i) \frac{i-1}{m}}-e^{- (y^* + (1-y^*) f_i) \frac{i}{m}} }{y^* + (1-y^*) f_i} \Bigg),
\end{align*}
where the inequality follows because $F(t) - t$ is decreasing when $t < t^*$. 
Observe that the final expression is a function with variables $f_1,f_2,\ldots,f_m$.

% To simplify the tedious calculation, we do some minor relaxations on some of the components in the derivative constraints.
% For any positive constant $\beta$, we have:

% \noindent
% For any $s\in \set{1,2,\dots, k-1}$,
% \begin{align*}
%  \int_{t_s}^{t_{s+1}}  \beta\left( F(t)-t \right)e^{-y^*  t -(1-y^*)F(t)} \,dt &\leq \beta(F_{s}-t_s) \int_{t_s}^{t_{s+1}}   e^{-y^* t -(1-y^*)F(t)} \,dt\\
%   & =\beta\frac{(F_{s}-t_s) e^{-(1-y^*) (F_s -t_s f_s)}}{-\left( y^*+(1-y^*) f_s \right)} e^{-t\left( y^*+(1-y^*) f_s \right) }\big\vert_{t_s}^{t_{s+1}}.
% \end{align*} 

% For any $s\in \set{k,k+1,\dots, m}$,
% \begin{align*}
%   & \quad \int_{t_s}^{t_{s+1}} \beta \left( z(t)-t \right) e^{-y^* t -(1-y^*)z(t)} \,dt\\
%   & \leq \int_{t_s}^{t_{s+1}} \beta\left( e^{-F(1)}\left( F(t^*)-t^*\right)+(1-e^{-F(1)}) (F_{s+1}-t_{s+1})\right) e^{-y^* t -(1-y^*)z(t)} \,dt\\
%   & = \frac{\beta\left( e^{-F(1)}\left( F(t^*)-t^*\right)+(1-e^{-F(1)}) (F_{s+1}-t_{s+1})\right) e^{-(1-y^*)\cdot \beta_1}}{-\left(y^*+(1-y^*)\left((1-e^{-F(1)}) f_s +e^{-F(1)}\right)\right)} e^{ -t  \left(y^*+(1-y^*)\left((1-e^{-F(1)}) f_s +e^{-F(1)}\right)\right)}  \Big\vert_{t_s}^{t_{s+1}},
% \end{align*} 
% where $\beta_1 = e^{-F(1)}\cdot (F(t^*)-t^*) + (1-e^{-F(1)})\cdot (F_s-f_st_s)$.

\subsection{Computational Results}\label{ssec: computation}

With the arguments and discretization methods we have introduced in the previous two subsections, deriving a lower bound on the competitive ratio is now straightforward:
it remains to use a computer program to enumerate the parameters $f_1,f_2,\ldots,f_m$ such that the discretized version of constraints~\eqref{eqn:cons1},~\eqref{eqn:cons2} and that $F(1)\geq 1$ hold, and $\min\{ r_1(y^*),r_2(y^*) \}$ is as large as possible, where $y^* = 1- \ln 2$.
Our experiment shows that we have $\min\{ r_1(y^*),r_2(y^*) \} \geq 0.6503$ when fixing the number of intervals to be $m = 40$, and the piecewise constant activation function as follows:
\begin{align*}
f(x)=\left\{
\begin{array}{ll}
    0, & \quad 0\leq x< 0.05,\\
    0.4, &  \quad 0.05\leq x <0.075, \\
    1, & \quad 0.075\leq x < 0.675,\\
    1.2, & \quad 0.675 \leq x<0.7,\\
    2, & \quad 0.7\leq x\leq 1.\\
\end{array}
\right.    
\end{align*}

It can also be verified that under the above activation function $f$, the discretized version of constraints~\eqref{eqn:cons1},~\eqref{eqn:cons2} and that $F(1)\geq 1$ hold. 
Hence we have the following lemma.

\begin{lemma}
    There exists a non-decreasing activation function $f$ with $F(1) \geq 1$ such that
     \begin{equation*}
         \min_{y \in [0, y^*]} \set{r_1(y)} \geq 0.650 \quad \text{and} \quad
         \min_{y \in [0, y^*]} \set{r_2(y)} \geq 0.650.
     \end{equation*}
\end{lemma}

Together with Lemma~\ref{lemma:analysis-framework} and the lower bounds we derived in Section~\ref{sec:final-lower-bounds}, we prove Theorem~\ref{theo:competitive-ratio}.

\section{Conclusions and Discussions}

In this work, we study the edge-weighted online stochastic matching problem and propose the \textsf{Evolving Suggested Matching} (ESM) algorithm that improves the state-of-the-art approximation ratio from $0.645$~\cite{Y22} to $0.650$.
The ESM is equipped with an activation function and with different choices of function the algorithm subsumes many existing algorithms for the online stochastic matching problem.
% On the other hand, our analysis captures the matching events in a much more fine-grained way compared to the one in~\cite{Y22}.
% 
% Similar to the current extensive works on online stochastic matching problems, our algorithm and analysis lie in the framework that treats any edge in a symmetric way. 
% We believe that our approach can not obtain competitive ratios significantly better than 0.650. 
A natural open question is to further improve this approximation ratio and explore the upper bounds of the best-possible approximation ratio for the edge-weighted online stochastic matching problem.
Our work (and \cite{Y22}) follows the framework of lower bounding the probability that some arbitrarily fixed edge is matched by the algorithm (see Lemma~\ref{lemma:analysis-framework}).
This framework enables us to simplify the analysis because we do not need to look at the edge weights (except that the optimal solution $\boldsymbol{x}$ depends on them).
On the other hand, it is unclear whether such restriction will prevent us from obtaining much better competitive ratios.
We believe that it would be a very interesting open problem to explore algorithms and analysis beyond this analysis framework.

\bibliography{refs}
\bibliographystyle{abbrv}

\end{document}